\documentclass[11pt]{article}
\usepackage{longtable}
\usepackage{ae,aecompl,pslatex}
\usepackage[dvips]{graphicx}
\usepackage{enumerate,color}
\usepackage{fullpage}
\usepackage{lineno}

\usepackage{placeins}
\usepackage{psfrag,epsfig,graphicx,color,array,dcolumn,tabularx,delarray,longtable,indentfirst,amsmath,amsfonts,amssymb,amsthm,eucal,bbm}
\theoremstyle{bkaexa} 
\usepackage{caption}
\usepackage{setspace}
\theoremstyle{bkaexa} 
\newtheorem{Rem}{Remark}
\theoremstyle{bkathm} 

\theoremstyle{bkathm} 
\newtheorem{Thm}{Theorem}
\theoremstyle{bkathm} 

\theoremstyle{bkathm} 
\newtheorem{Lem}{Lemma}
\theoremstyle{definition}

\begin{document}
\setstretch{1.5}
\title{On relationships between Chatterjee's and Spearman's correlation coefficients}
\author{\normalsize Qingyang Zhang\\
\normalsize Department of Mathematical Sciences, University of Arkansas, AR 72701\\
\normalsize Email: qz008@uark.edu
}
\date{}
\maketitle

\begin{abstract}
In his seminal work, Chatterjee (2021) introduced a novel correlation measure which is distribution-free, asymptotically normal, and consistent against all alternatives. In this paper, we study the probabilistic relationships between Chatterjee's correlation and the widely used Spearman's correlation. We show that, under independence, the two sample-based correlations are asymptotically joint normal and asymptotically independent. Under dependence, the magnitudes of two correlations can be substantially different. We establish some extremal cases featuring large differences between these two correlations. Motivated by these findings, a new independence test is proposed by combining Chatterjee's and Spearman's correlations into a maximal strength measure of variable association. Our simulation study and real data application show the good sensitivity of the new test to different correlation patterns.
\end{abstract}


\section{Introduction}
Measuring and testing the dependence between two continuous variables is a durable research topic in statistics. Two classical and arguably most widely used dependence measures are Pearson's correlation and Spearman's correlation. Pearson's correlation is powerful in detecting linear dependence, especially when the two variables are bivariate normal. Spearman's correlation is a nonparametric alternative to Pearson's. It is sensitive to monotonic relations, and generally robust to outliers since it is rank-based. Under the null hypothesis of independence, the two sample-based correlations are both asymptotically normal, making it easy to calculate p-values. However, the common drawback of these methods is that they generally fail to detect non-monotonic relationships. 

In the past decades, there have been numerous tests developed that are consistent against all alternatives, including the kernel-based test \cite{hsic}, distance correlation test \cite{dcor}, sign covariance test \cite{BD}, copula-based test \cite{SW}, graph-based test \cite{graph} and maximal information test \cite{mic}, among many others. For a recent survey, see Josse \& Holmes (2016) \cite{survey}. Some of these tests are popular among practitioners, e.g., the distance correlation test. However, one major bottleneck of these tests is the testing process: because there is a lack of simple asymptotic theory that facilitate analytical computation of p-values, an expensive permutation test is typically required. For instance, the asymptotic null distribution of distance correlation is difficult to derive because it depends on the underlying distributions of random variables, and the standard approach is to approximate the null distribution of distance covariance via permutation, which requires a time complexity of $O(Rn^2)$, where $R$ is the number of permutations and $n$ is the sample size. 

Recently, Chatterjee (2021) introduced a rank-based correlation test which is also consistent against all alternatives \cite{chatterjee}. Different from the aforementioned tests, Chatterjee's correlation is asymptotically normal under independence, facilitating quick computation of p-values. Due to its nice properties, Chatterjee's correlation has attracted much attention over the past two years. We begin with a brief review of this method and related literature. Let $X$ and $Y$ be two continuous variables, and $(X_{i}, Y_{i})_{i=1,...,n}$ be n $i.i.d.$ samples of $(X, Y)$. Assuming that $X_{i}$'s and $Y_{i}$'s have no ties, the data can be uniquely arranged as $(X_{(1)}, Y_{(1)}), ..., (X_{(n)}, Y_{(n)})$, such that $X_{(1)}<\cdots<X_{(n)}$. Here $Y_{(1)}, ..., Y_{(n)}$ denote the concomitants. Let $R_{i}$ be the rank of $Y_{(i)}$, i.e., $R_{i} = \sum_{k=1}^{n}\mathbbm{1}\{Y_{(k)}\leq Y_{(i)}\}$, Chatterjee's correlation $\xi_{n}(X, Y)$ is defined as
\begin{equation}
\xi_{n}(X, Y) = 1-\frac{3\sum_{i=1}^{n-1}|R_{i+1}-R_{i}|}{n^2-1}.
\end{equation}
 
The asymptotic behavior of $\xi_{n}(X, Y)$ and related problems have been examined in recent papers. Here we outline a few of them which are most relevant to this work. In his original paper \cite{chatterjee}, Chatterjee showed that 
\begin{equation*}
\xi_{n}(X, Y)\rightarrow \xi(X, Y) = \frac{\int V(E(\mathbbm{1}\{Y\geq t|X\}))dF_{Y}(t)}{\int V(\mathbbm{1}\{Y\geq t\})dF_{Y}(t)},  ~\mbox{as}~n\rightarrow\infty.
\end{equation*}
The limiting quantity $\xi(X, Y)$ is also known as Dette-Siburg-Stoimenov's dependence measure \cite{dss}, which is between 0 and 1 (0 if and only if $X$ and $Y$ are independent, 1 if and only if $Y$ is a measurable function of $X$). Chatterjee (2021) also established the asymptotic normality of $\xi_{n}(X, Y)$ under independence. Precisely, $\sqrt{n}\xi_{n}(X, Y)\xrightarrow[]{d} N(0, 2/5), ~\mbox{as}~n\rightarrow\infty$. The central limit theorem of $\xi_{n}(X, Y)$ under dependence (as long as $Y$ is not a measurable function of $X$) is recently proved by Lin \& Han (2022) \cite{LinHan22}. In addition, Auddy et al. (2021) investigated the limiting power of $\xi_{n}(X, Y)$ under local alternatives and obtained the exact detection threshold \cite{auddy}. The fast growing literature on Chatterjee's correlation also include Shi et al. (2021a), Shi et al. (2021b), Lin \& Han (2021), Cao \& Bickel (2020), Deb et al. (2020), Han \& Huang (2022), Azadkia \& Chatterjee (2021), Zhang (2023), Chatterjee \& Vidyasagar (2022), among many others.

In addition to the simple asymptotic theory, as an empirical finding, Chatterjee's test is powerful to detect non-monotonic associations, especially those with oscillating nature such as the W-shaped scatterplot and the sinusoid \cite{chatterjee}. The only disadvantage of Chatterjee's test is that it may have less statistical power for smoother alternatives (such as linear or other monotonic relationships), compared to other popular tests including distance correlation test and Bergsma-Dassios test. For instance, as shown in Figure 5 of Chatterjee (2021), the power of $\xi_{n}(X, Y)$ quickly deteriorates as the noise level increases, which could be a matter of concern in practice. Motivated by these facts, we propose a versatile test by taking the maximum of Chatterjee's correlation and Spearman's correlation, where the the latter one is powerful to detect monotonic and smoother associations. Two questions arising from this proposal are 

\begin{itemize}
\item[(1)] What is the asymptotic joint distribution of the two correlations under independence? 
\item[(2)] How much can they differ as a measure of dependence? 
\end{itemize}

The first question is about analytical calculation of p-values. The second question investigates if the two correlations to be combined are complementary, in the sense that they measure different dependencies. In this paper, we give a complete answer to the first question. For the second question, we provide two extremal examples featuring large differences between the two correlations. The idea of combining two complementary correlation metrics is not new. For instance, Zhang, Qi \& Ma (2011) showed the asymptotic independence between Pearson's correlation and a quotient-type correlation, and proposed a new test by combining them into a maximal type measure \cite{ZhangQi2011}.

The remainder of this paper is structured as follows: Section 2 derives the asymptotic joint distribution of $S_{n}(X, Y)$ and $\xi_{n}(X, Y)$ under independence. Section 3 investigates how much the two correlations can differ under dependence. Section 4 proposes the new test of independence, validated by both synthetic data and a real-world dataset. Section 5 discusses the paper with some future perspectives.

\section{Asymptotic joint distribution under independence}
With the same notations in previous section, Spearman's rank correlation can be written as 
\begin{equation*}
S_{n}(X, Y) = 1-\frac{6\sum_{i=1}^{n}(i-R_{i})^2}{n(n^2-1)},
\end{equation*}
where $R_{i}$ represents the rank of $Y_{(i)}, i = 1, ..., n$. Under the hypothesis of independence, it is well known that $E(\sqrt{n}S_{n}(X, Y)) = 0$, $V(\sqrt{n}S_{n}(X, Y)) = n/(n-1)$, and $\sqrt{n}S_{n}(X, Y)\xrightarrow[]{d} N(0, 1), ~\mbox{as}~n\rightarrow\infty$. Though $\xi_{n}(X, Y)$ and $S_{n}(X, Y)$ are both asymptotically normal, their joint behavior remains unexplored. In this section, we derive the asymptotic joint distribution of $\xi_{n}(X, Y)$ and $S_{n}(X, Y)$ under independence. 

Let $[n] := \{1, 2, ..., n\}$ be the sample indices. Under independence, $\{R_{1}, ..., R_{n}\}$ is a random permutation of $[n]$. We first show that $\xi_{n}(X, Y)$ and $S_{n}(X, Y)$ are uncorrelated for finite sample, as stated in the following lemma:
\begin{Lem}
If $X$ and $Y$ are independent, we have
\begin{equation*}
\mbox{Cov}\left[ S_{n}(X, Y), \xi_{n}(X, Y) \right ] = 0,
\end{equation*}
for any $n\geq 2$. 
\end{Lem}
\begin{proof}
Spearman's correlation can be rewritten as
$$S_{n}(X, Y) = -\frac{3(n+1)}{n-1}+\frac{12\sum_{i=1}^{n}iR_{i}/n}{(n^2-1)}.$$
For the covariance between $\xi_{n}(X, Y)$ and $S_{n}(X, Y)$, we have
\begin{align}
& \mbox{Cov}\left[\sum_{i=1}^{n-1}|R_{i+1}-R_{i}|, \sum_{j=1}^{n}\frac{j}{n}R_{j}\right]  \nonumber\\
= & \sum_{i=1}^{n-1}\sum_{j=1}^{n}\mbox{Cov}\left[R_{i+1}+R_{i}-2\min(R_{i+1}, R_{i}), \frac{j}{n}R_{j}\right] \nonumber\\
= & \sum_{i=1}^{n-1}\sum_{j=1}^{n}\frac{j}{n} \mbox{Cov}\left[R_{i+1}, R_{j}\right] + \sum_{i=1}^{n-1}\sum_{j=1}^{n}\frac{j}{n} \mbox{Cov}\left[R_{i}, R_{j}\right] - 2\sum_{i=1}^{n-1}\sum_{j=1}^{n}\frac{j}{n} \mbox{Cov}\left[\min(R_{i+1}, R_{i}), R_{j}\right]. 
\end{align}
The following results (\cite{LinHan}, Lemma 6.1, page 13) are needed for our derivations
\begin{align*}
\mbox{Cov}[R_{1}, R_{2}] & = -\frac{n+1}{12} \\
\mbox{V}[R_{1}] & = \frac{(n-1)(n+1)}{12} \\
\mbox{Cov}[R_{1}, \min(R_{1}, R_{2})] & = \frac{(n+1)(n-2)}{24} \\
\mbox{Cov}[R_{1}, \min(R_{2}, R_{3})] & = -\frac{n+1}{12}. 
\end{align*}
For the first term in Equation (2), we have
\begin{align*}
\sum_{i=1}^{n-1}\sum_{j=1}^{n} \frac{j}{n} \mbox{Cov}[R_{i+1}, R_{j}] & = \sum_{i=1}^{n-1} \left\{ \frac{i+1}{n}\mbox{V}[R_{1}] + \sum_{j\neq i+1}\frac{j}{n}\mbox{Cov}[R_{1}, R_{2}] \right\}\\
& = \sum_{i=1}^{n-1} \left\{ \frac{(n+1)(i+1)}{12} - \frac{(n+1)^2}{24} \right\} \\
& = \frac{(n+1)(n-1)}{24}
\end{align*}
Similarly for the second term, we have
$$\sum_{i=1}^{n-1}\sum_{j=1}^{n}\frac{j}{n} \mbox{Cov}\left[R_{i}, R_{j}\right] = -\frac{(n+1)(n-1)}{24}.$$
For the third term, we have 
\begin{align*}
\sum_{i=1}^{n-1}\sum_{j=1}^{n}\frac{j}{n} \mbox{Cov}\left[\min(R_{i+1}, R_{i}), R_{j}\right] & = \sum_{i=1}^{n-1} \left\{ \frac{2i+1}{n}\mbox{Cov}[R_{1}, \min(R_{1}, R_{2})] + \sum_{j\neq i, i+1}\frac{j}{n}\mbox{Cov}[R_{1}, \min(R_{2}, R_{3})] \right\}\\
& = \sum_{i=1}^{n-1} \left\{ \frac{(2i+1)(n+1)(n-2)}{24n} - \frac{(n+1)^2}{24} + \frac{(2i+1)(n+1)}{12n} \right\} \\
& = 0.
\end{align*}
Therefore $\mbox{Cov}\left[ S_{n}(X, Y), \xi_{n}(X, Y) \right ] = 0$. This completes the proof of Lemma 1.
\end{proof}

\begin{Rem}
It is noteworthy that Lemma 1 only indicates the uncorrelatedness between $S_{n}(X, Y)$ and $\xi_{n}(X, Y)$. In fact, under finite sample, $S_{n}(X, Y)$ and $\xi_{n}(X, Y)$ are generally dependent. A simple example is given in Table 1, where $n=3$ and $\mbox{Cov}\left[ |S_{3}(X, Y)|, \xi_{3}(X, Y) \right ] = 1/24$.

\begin{table}[!htbp]
\caption{A special case when $n=3$}
\centering
\begin{tabular}{l r r r r}
\hline\hline
($R_{1}$, $R_{2}$, $R_{3}$) & $\xi_{3}(X, Y)$ & $S_{3}(X, Y)$ & $|S_{3}(X, Y)|$\\ [0.5ex]
\hline
(1, 2, 3) & 1/4 & 1 & 1 \\
(1, 3, 2) & -1/8 & 1/2 & 1/2  \\
(2, 1, 3) & -1/8 & 1/2 & 1/2  \\
(2, 3, 1) & -1/8 & -1/2 & 1/2 \\
(3, 1, 2) & -1/8 & -1/2 & 1/2 \\
(3, 2, 1) & 1/4 & -1 & 1 \\ [1ex]
\hline
\end{tabular}
\end{table}
\end{Rem}

Next we present a lemma that establishes the central limit theorem for $\{S_{n}(X, Y), \xi_{n}(X, Y)\}$. The key steps to prove Lemma 2 include (1) the coupling method for permutation oscillation proposed by \cite{angus} (2) the central limit theorem for m-dependent sequence and (3) Cra\'mer-Wold device. The detailed proof is a bit lengthy, and we provide it in the Appendix. 

\begin{Lem}
If $X$ and $Y$ are independent, $\sqrt{n}S_{n}(X, Y)$ and $\sqrt{n}\xi_{n}(X, Y)$ are asymptotically joint normal.
\end{Lem}

By Lemmas 1 and 2, our main theorem follows immediately

\begin{Thm}
If $X$ and $Y$ are independent, $\sqrt{n}S_{n}(X, Y)$ and $\sqrt{n}\xi_{n}(X, Y)$ are asymptotically joint normal and asymptotically independent. To be specific,
$$\begin{bmatrix}
\sqrt{n}S_{n}(X, Y)\\
\sqrt{n}\xi_{n}(X, Y) 
\end{bmatrix}
\xrightarrow{d} N\left [
\begin{pmatrix}
0\\
0 
\end{pmatrix}, 
\begin{pmatrix}
1 & 0\\
0 & 2/5
\end{pmatrix} \right ]$$ as $n\rightarrow\infty$.
\end{Thm}

Theorem 1 answers the first question that we asked in Section 1, which enables analytical calculation of p-values for the proposed integrated test (to be further discussed in Section 4). Theorem 1, together with Lemma 1 and Remark 1, give a complete characterization for the joint behavior of $S_{n}(X, Y)$ and $\xi_{n}(X, Y)$ under independence. The convergence of $\{S_{n}(X, Y), \xi_{n}(X, Y)\}$ to joint normality, as sample size $n$ increases, is illustrated in the Figure 1 below. 

\begin{center}
[Figure 1 about here]
\end{center}

\section{Chatterjee's and Spearman's correlations - how much can they differ?}
In this section, we explore the second question outlined in the Section 1, i.e., how much the two correlations can differ as a metric of dependence. We focus on some extremal cases where the magnitudes of $\xi_{n}(X, Y)$ and $S_{n}(X, Y)$ are largely different. For dependent variables $X$ and $Y$, $\xi_{n}(X, Y)$ is generally, though not always, between 0 and 1, while $S_{n}(X, Y)$ is between -1 and 1, therefore we take the absolute value of $S_{n}(X, Y)$, and compare $\xi_{n}(X, Y)$ and $|S_{n}(X, Y)|$ instead. We first provide an extremal case where the absolute Spearman's correlation is small but Chatterjee's correlation is large. This extremal case is easy to construct using simple symmetric patterns such as $Y=|X|$ or $Y=X^2$, $-1<X<1$. 

\noindent
\textbf{Extremal case 1}: For any $\epsilon>0$, there exist ranks $\{R_{1}, ..., R_{n}\}$, such that $|S_{n}(X, Y)| < \epsilon$ and $\xi_{n}(X, Y)>1-\epsilon$.

\begin{proof}
Without loss of generality, suppose $n$ is odd. We construct the following ranks 
$$
R_{i} = \begin{cases}
			n-2(i-1), & 1\leq i \leq (n+1)/2\\
            2i-(n+1), & (n+3)/2 \leq i \leq n.
		 \end{cases}
$$
It is straightforward to show that 
\begin{equation*}
\xi_{n}(X, Y) = 1-\frac{6n-9}{n^2-1}.
\end{equation*}
To derive $S_{n}(X, Y)$, first we have 
\begin{align*}
\sum_{i=1}^{n}(i-R_{i})^2 & = \sum_{i=1}^{(n+1)/2}(i-R_{i})^2 +  \sum_{i=(n+3)/2}^{n}(i-R_{i})^2 \\
& = \frac{(n-1)(n+1)(n+2)}{8} + \frac{n(n+1)(n-1)}{24} \\ 
& = \frac{(n-1)(n+1)(2n+3)}{12},
\end{align*}
then
\begin{equation*}
|S_{n}(X, Y)| = \frac{3}{2n}.
\end{equation*}
For any given $\epsilon>0$, we can find an odd number $n$, such that $(6n-9)/(n^2-1)<\epsilon$ and $3/(2n)<\epsilon$, therefore 
$|S_{n}(X, Y)| < \epsilon$ and $\xi_{n}(X, Y)>1-\epsilon$.
\end{proof}

Next we seek an opposite extremal case where $|S_{n}(X, Y)|$ is large but $\xi_{n}(X, Y)$ is small. This extremal case is not straightforward because when $|S_{n}(X, Y)| = 1$, $\{R_{1}, ..., R_{n}\}$ are monotonically increasing or decreasing, therefore $\xi_{n}(X, Y) = (n-2)/(n+1)$, which means when $|S_{n}(X, Y)|$ is close to 1, the minimum possible value of $\xi_{n}(X, Y)$ may not be close to 0. Mathematically, this can be formulated as the following optimization problem

For a given $0<\epsilon<1$, find 
$$\max_{\{R_{1}, ..., R_{n}\} \overset{\mathrm{perm}}{=} [n]}\left|1-\frac{6\sum_{i=1}^{n}(i-R_{i})^2}{n(n^2-1)}\right|,$$ 
where $\{R_{1}, ..., R_{n}\} \overset{\mathrm{perm}}{=} [n]$ represents that $\{R_{1}, ..., R_{n}\}$ is a permutation of $\{1, ..., n\}$, given the following inequality constraint
$$1-\frac{3\sum_{i=1}^{n-1}|R_{i+1}-R_{i}|}{n^2-1}<\epsilon.$$

Unless $n$ is small enough to enumerate all permutations of $\{R_{1}, ..., R_{n}\}$, the optimization problem above is difficult because of the complicated constraints. It may require advanced integer programming techniques, which is beyond the scope of this work. We leave the optimization problem as an open question and try to give a simple example instead, where $\xi_{n}(X, Y)$ is relatively small but $|S_{n}(X, Y)|$ is substantially larger. Our intuition is that Spearman's correlation measures the overall monotonic relations while Chatterjee's correlation is sensitive to the local changes. Accordingly we construct a case which is overall monotonic but with wiggly local patterns. 

\noindent
\textbf{Extremal case 2}: For any $\epsilon>0$, there exist ranks $\{R_{1}, ..., R_{n}\}$, such that $\xi_{n}(X, Y) = \epsilon + O(1/n)$ and $|S_{n}(X, Y)| = 1- \sqrt{2/27}(1-\epsilon)^{3/2} + O(1/n)$.

\begin{proof}
We construct $n = 2m+p$ ranks which can be partitioned into two parts, the part of $1\leq i\leq 2m$ has an oscillating pattern while the part of $2m+1\leq i\leq n$ is monotonically increasing
$$
R_{i} = \begin{cases}
			(i+1)/2, & 1\leq i \leq 2m \text{   and } i \text{ is odd}\\
            i/2 + m, & 1\leq i \leq 2m \text{   and } i \text{ is even}\\
            i, & 2m+1\leq i \leq n. 
		 \end{cases}
$$
Let $c=p/m$, the following results can be obtained
\begin{align*}
\xi_{n}(X, Y) & = 1- \frac{3[m^2+(m-1)^2+p]}{(2m+p)^2-1} \\
& = 1-\frac{6}{(c+2)^2} + O(1/n),
\end{align*}
\begin{align*}
|S_{n}(X, Y)| & = 1 - \frac{2m(m+1)(2m+1)-6m^2}{[(2m+p)^2-1](2m+p)} \\
& = 1 - \frac{4}{(c+2)^3} + O(1/n).
\end{align*}
For any $1>\epsilon>0$, there exists $c>0$ such that $\epsilon = 1-6/(c+2)^2$, therefore $\xi_{n}(X, Y) = \epsilon + O(1/n)$. By the same $c$, we have 
$|S_{n}(X, Y)| = 1- \sqrt{2/27}(1-\epsilon)^{3/2} + O(1/n)$.
\end{proof}
We give two examples for this extremal case (1) when $n = 100$, $m = 40$ and $p = 20$,  $\xi_{n}(X, Y) \approx 0.058$ while $S_{n}(X, Y) \approx 0.753$ (2) when $n = 130$, $m = 50$ and $p = 30$,  $\xi_{n}(X, Y) \approx 0.125$ while $S_{n}(X, Y) \approx 0.779$, both showing substantial difference between the two metrics.

\section{A new test for independence}
Motivated by the findings in Sections 2 and 3, we propose the following new metric 
\begin{equation*}
I_{n}(X, Y) = \max\{|S_{n}(X, Y)|, \sqrt{5/2}\xi_{n}(X, Y)\}.
\end{equation*}

As $I_{n}(X, Y)$ takes advantage of both $S_{n}(X, Y)$ and $\xi_{n}(X, Y)$, it can be used as a versatile test for detecting both monotonic and non-monotonic associations. Moreover, by Theorem 1, one can calculate the asymptotic p-value as follows
\begin{equation*}
P(\sqrt{n}I_{n}(X, Y) > z) \approx 1 - \Phi(z) \left[1-2\Phi(-z)\right],
\end{equation*}
where $z\geq 0$ and $\Phi(\cdot)$ represents the standard normal distribution function. 

We conducted two simulation studies to evaluate the performance of the proposed test. In the first study, we compared the empirical power of $S_{n}(X, Y)$, $\xi_{n}(X, Y)$ and $I_{n}(X, Y)$ under different sample sizes $\{20, 40, 60, 80, 100\}$. Generating $X$ from $\mbox{Uniform}[-1, 1]$, the following four alternatives were considered, where $\epsilon\perp X$ and $\epsilon\sim N(0,1)$:

\begin{itemize}
\item[1.] Linear: $Y = X+\epsilon$
\item[2.] Quadratic: $Y = X^2+0.3\epsilon$.
\item[3.] Sinusoid: $Y = \cos(2\pi X)+0.75\epsilon$.
\item[4.] Stepwise: $Y = \mathbbm{1}_{\{-1\leq X\leq -0.5\}}+2*\mathbbm{1}_{\{-0.5<X\leq0\}}+3*\mathbbm{1}_{\{0<X\leq 0.5\}}+4*\mathbbm{1}_{\{0.5<X\leq 1\}}+2\epsilon$.
\end{itemize} 

Figure 2 summarizes the empirical power over $5,000$ simulation runs (at the significance level of 0.05). As expected, Spearman's test has the highest power for the monotonic settings, i.e, linear and stepwise, but extremely low power for the other settings. Chatterjee's test is most powerful for two non-monotonic settings, i.e., quadratic and sinusoid, but it has much lower power for the monotonic settings. For instance, in the linear setting when $n=60$, Chatterjee's test has a power of 0.532, while the other two tests both have power higher than 0.98. The new test has satisfactory power for all settings, especially for linear and stepwise settings where the new test is comparable to Spearman's method.

\begin{center}
[Figure 2 about here]
\end{center}

In the second study, we examined the p-value bias. The exact p-value was approximated using $5,000$ permutations and the bias was computed as the asymptotic p-value minus the exact p-value. In each simulation run, we generated $X$ from $\mbox{Uniform}[-1, 1]$ and $Y$ from $N(0,1)$ with sample size $\{20, 40, 60, 80, 100\}$. Figure 3 summarizes the bias over $1,000$ simulations runs. It can be seen that the asymptotic p-values are overall close to the exact p-values. However, for a relatively small sample size, e.g., $n=20$, the asymptotic p-values is positively biased, indicating the conservativeness of the test. The bias vanishes as sample size increases. In practice, if the sample size is small, e.g., $n<30$, we recommend a permutation test based on $I_{n}(X, Y)$ to avoid power loss.

\begin{center}
[Figure 3 about here]
\end{center}

The proposed method was also tested on a transcriptomics dataset by Spellman et al. (1998), which contains the expression levels of 6,223 yeast genes over 23 successive time points during the cell cycle \cite{spellman}. This dataset was processed by \cite{mic}, where genes with missing observations were excluded. The processed dataset has 4,381 genes, which is available through R package \textit{minerva}. There have been many papers testing different correlation measures using this particular dataset including Chatterjee (2021). 

For all three methods, p-values were calculated using asymptotic formulas, and then adjusted by Benjamini-Hochberg procedure to control the false discovery rate (FDR) at the level of 0.05. Figure 4 summarizes the number of significant genes identified by three tests. Out of a total of 4,381 genes, the new test selected 734 genes whose expression levels change during the cell cycle, while the other two tests selected 619 and 385 genes respectively. This is due to the existence of different expression  patterns in the data, i.e., some genes have smoother expression change while others have non-monotonic such as oscillating expression change. Figure 5 presents a random sample of four genes that are identified by the new test but missed by Spearman's test. It can be seen that the expression levels of all four genes exhibit certain oscillating patterns. Figure 6 shows a random sample of four genes that are identified by the new test but missed by Chatterjee's test, where all genes have smoother expression change during the cell cycle. 

\begin{center}
[Figure 4 about here]
\end{center}

\begin{center}
[Figure 5 about here]
\end{center}

\begin{center}
[Figure 6 about here]
\end{center}

\section{Discussion and conclusions}
Chatterjee's rank correlation has attracted a lot of attention during the past two years due to its simplicity and nice statistical properties. However, the cost we pay for this simple method is its inferior performance in detecting smoother correlation patterns, such as linear relationships. To boost the power of this ingenious measure, in this paper, we proposed a max-type test by combining Chatterjee's correlation with Spearman's correlation, as the latter one is also rank based but sensitive to smooth correlation patterns. We derive the asymptotic joint distribution of these two correlations under independence, which enables analytical calculation of p-values. Our simulation study and the transcriptomics application illustrated the promise of the new test. Due to the simple calculation and satisfactory performance, the test is readily applicable to many correlative analyses, e.g., the gene-gene interaction and protein-protein interaction network construction.

There are several possible extensions of this work. First, the new test statistic $I_{n}(X, Y)$ is generally asymmetric because $\xi_{n}(X, Y)$ is asymmetric, i.e., $\xi_{n}(X, Y)\neq \xi_{n}(Y, X)$. When a symmetric measure is more suitable, one can consider the following modification
$$I^{sym}_{n}(X, Y) = \max\{S_{n}(X, Y), \xi_{n}(X, Y) , \xi_{n}(Y, X) \}.$$ 
In previous work \cite{Zhang23}, we established the asymptotic joint normality of $ \xi_{n}(X, Y)$ and $\xi_{n}(Y, X)$ and showed that the symmetrized metric, i.e., $\max\{ \xi_{n}(X, Y), \xi_{n}(Y, X)\}$, converges to a skew normal distribution under independence. The proof is based on Chatterjee's central limit theorem \cite{chatterjee.clt}. The joint asymptotic behavior of $\{S_{n}(X, Y), \xi_{n}(X, Y) , \xi_{n}(Y, X)\}$ could be studied in a similar way, and the first and most important step is to construct a valid interaction rule for $I^{sym}_{n}(X, Y)$ \cite{auddy, Zhang23}. However, this may require significant efforts and we leave it for future research. 

Second, as we discussed in the simulation study, the asymptotic p-value is generally close to the true p-value, but it tends to be positively biased for small sample, e.g., $n<30$, resulting in certain power loss. In the case of small sample, we recommend a permutation test for better testing performance. Another way to reduce the potential p-value bias is to use asymptotic expansion method, e.g., Edgeworth expansion, Cornish-Fisher expansion or saddle point approximation, which may improve p-value approximation by incorporating higher-order moments such as skewness and kurtosis.

\section{Appendix: Proof of Lemma 2}
\begin{proof}
We first define $F(y) = P(Y<y)$, $U_{i} = F(Y_{(i)})$, $F_{n}(y) = \sum_{i=1}^{n}\mathbbm{1}\{Y_{(i)}\leq y\}/n$, and $F_{n}(x) = \sum_{i=1}^{n}\mathbbm{1}\{U_{i}\leq x\}/n$. For Chatterjee's correlation, using Equations 5-8 in Angus (1995), we have 
\begin{equation*}
\frac{\sum_{i=1}^{n-1}|R_{i+1}-R_{i}|-n(n-1)/3}{\sqrt{n}(n-1)} = \frac{1}{\sqrt{n}}\sum_{i=1}^{n-1}\left[ |U_{i+1} - U_{i}| + 2U_{i}(1-U_{i})-\frac{2}{3} \right] + Z,
\end{equation*}
where $Z\xrightarrow{P}0$. For Spearman's correlation, we define the following function
\begin{equation*}
G_{n}(x) = \frac{1}{n}\sum_{i=1}^{n}\frac{2i}{n+1}\mathbbm{1}\{U_{i}\leq x\}.
\end{equation*}
Since
\begin{equation*}
\frac{1}{n}\sum_{i=1}^{n}\frac{2i}{n+1}\mathbbm{1}\{U_{i}\leq x\}\leq  \frac{1}{n}\sum_{i=1}^{n}\frac{2i}{n+1} = 1,
\end{equation*}
we have $0\leq G_{n}(x)\leq 1$. The expectation and variance of $G_{n}(x)$ are $E\left[G_{n}(x) \right]= x$ and 
\begin{equation*}
V\left[G_{n}(x)\right] = \frac{2x(1-x)(2n+1)}{3n(n+1)}\leq \frac{2n+1}{6n(n+1)} \rightarrow 0,
\end{equation*}
therefore $G_{n}(x)\xrightarrow{P}x$ for $x\in[0, 1]$, as $n\rightarrow \infty$. It is also noteworthy that
\begin{align*}
\frac{1}{n\sqrt{n}}\left(\sum_{i=1}^{n}\frac{2i}{n+1}R_{i} - \frac{n^2}{2}\right)& = \int \sqrt{n}\left[ F_{n}(x) - \frac{1}{2} \right] d G_{n}(x) \\
& =  \int \sqrt{n}\left[ F_{n}(x) - x \right] d G_{n}(x) + \int \sqrt{n}\left[ x - \frac{1}{2} \right] d G_{n}(x) 
\end{align*}
where the second term can be rewritten as
\begin{equation*}
\int \sqrt{n}\left[ x - \frac{1}{2} \right] d G_{n}(x) = \frac{1}{\sqrt{n}}\sum_{i=1}^{n}\frac{2i}{n+1}U_{i} - \frac{\sqrt{n}}{2}.
\end{equation*}
For the first term, using continuous mapping theorem, we have
\begin{align*}
 \int \sqrt{n}\left[ F_{n}(x) - x \right] d G_{n}(x) & \xrightarrow{d}  \int \sqrt{n}\left[ F_{n}(x) - x \right] dx \\
 & = \frac{\sqrt{n}}{2} - \frac{1}{\sqrt{n}}\sum_{i=1}^{n}U_{i},
\end{align*}
therefore
\begin{equation*}
\frac{1}{n\sqrt{n}}\left(\sum_{i=1}^{n}\frac{2i}{n+1}R_{i} - \frac{n^2}{2}\right) \xrightarrow{d} \sum_{i=1}^{n} \left[ \frac{2i}{n+1} -1 \right] U_{i}.
\end{equation*}
We will show that $\sum_{i=1}^{n-1}\left[ |U_{i+1} - U_{i}| + 2U_{i}(1-U_{i})-\frac{2}{3} \right]/\sqrt{n}$ and $\sum_{i=1}^{n} \left[ 2i/(n+1) -1 \right] U_{i}/\sqrt{n}$ are asymptotically joint normal.

For any two constants, $a$ and $b$, define
\begin{equation*}
Z_{i} =  a|U_{i+1} - U_{i}| + 2aU_{i}(C_{i}-U_{i}) -\frac{2a}{3}
\end{equation*}
and 
\begin{equation*}
W_{n} = \frac{1}{\sqrt{n}}\sum_{i=1}^{n-1}Z_{i},
\end{equation*}
where 
\begin{equation*}
C_{i} = 1+\frac{bi}{a(n+1)}-\frac{b}{2a}.
\end{equation*}
It can be seen that for any $j\geq 1$, $Z_{i+j}$ is independent of $[Z_{1}, ..., Z_{i}]$, therefore the sequence $\{Z_{i}\}$ is 1-dependent sequence. 
Similar to Equations (11)-(14) in \cite{angus}, we have
\begin{equation*}
\mbox{V}[Z_{i}] = \mbox{V}[a|U_{i+1} - U_{i}|] + \mbox{V}[2aU_{i}(C_{i}-U_{i})] + 2\mbox{Cov}[a|U_{i+1} - U_{i}|, 2aU_{i}(C_{i}-U_{i})],
\end{equation*}
where
\begin{equation*}
\mbox{V}[a|U_{i+1} - U_{i}|] = \frac{a^2}{18},
\end{equation*}
and
\begin{equation*}
\mbox{V}[2aU_{i}(C_{i}-U_{i})] = 4a^2\left[ \frac{C_{i}^2}{12} - \frac{C_{i}}{6} + \frac{4}{45} \right].
\end{equation*}
For the covariance term, it can be shown that $\mbox{Cov}[|U_{i+1} - U_{i}|, U_{i}] = 0$, therefore
\begin{align*}
2\mbox{Cov}[a|U_{i+1} - U_{i}|, 2aU_{i}(C_{i}-U_{i})] & = -4a^2\mbox{Cov}[|U_{i+1} - U_{i}|, U_{i}^2] \\
& = -\frac{a^2}{45},
\end{align*}
Summarizing the results above, we get
\begin{align*}
\mbox{V}[Z_{i}] & = \frac{a^2}{18} + 4a^2\left[ \frac{C_{i}^2}{12} - \frac{C_{i}}{6} + \frac{4}{45} \right] -\frac{a^2}{45}\\
 & = \frac{a^2}{18} + \frac{b^2}{12}\frac{(2i-n-1)^2}{(n+1)^2} \\
 & \geq \frac{a^2}{18} 
\end{align*}
For the covariance between $Z_{i}$ and $Z_{i+1}$, we have 
\begin{align*}
\mbox{Cov}[Z_{i}, Z_{i+1}] & = \mbox{Cov}[a|U_{i+1} - U_{i}| + 2aU_{i}(C_{i}-U_{i}), a|U_{i+2} - U_{i+1}| + 2aU_{i+1}(C_{i+1}-U_{i+1})] \\
 & = \mbox{Cov}[a|U_{i+1} - U_{i}|, a|U_{i+2} - U_{i+1}|] + \mbox{Cov}[2aU_{i}(C_{i}-U_{i}), 2aU_{i+1}(C_{i+1}-U_{i+1})]\\
 & = \frac{a^2}{180} - \frac{a^2}{90}\\
 & = -\frac{a^2}{180},
\end{align*}
therefore
\begin{align*}
\mbox{V}\left[\sum_{i=1}^{n-1}Z_{i}\right] & = \sum_{i=1}^{n-1}\mbox{V}(Z_{i}) + 2\sum_{i=1}^{n-1}\mbox{Cov}(Z_{i}, Z_{i+1}) \\
& \geq \frac{2(n-1)a^2}{45},
\end{align*}
and
$$\frac{\sqrt{\mbox{V}\left[\sum_{i=1}^{n-1}Z_{i}\right]}}{n^{1/3}}\rightarrow \infty$$ 
as $n\rightarrow\infty$. Using the central limit theorem for m-dependent random variables (Theorem 1 in \cite{angus}), $W_{n}$ converges in distribution to a normal distribution. Finally, by Cra\'mer-Wold device,
$\sqrt{n}S_{n}(X, Y)$ and $\sqrt{n}\xi_{n}(X, Y)$ are asymptotically joint normal. 
\end{proof}

 \section*{Competing Interests}
\noindent
The author has declared that no competing interests exist.

\newpage
\section*{Figures}
\begin{figure}[!htbp]
\begin{center}
\includegraphics[scale=0.7]{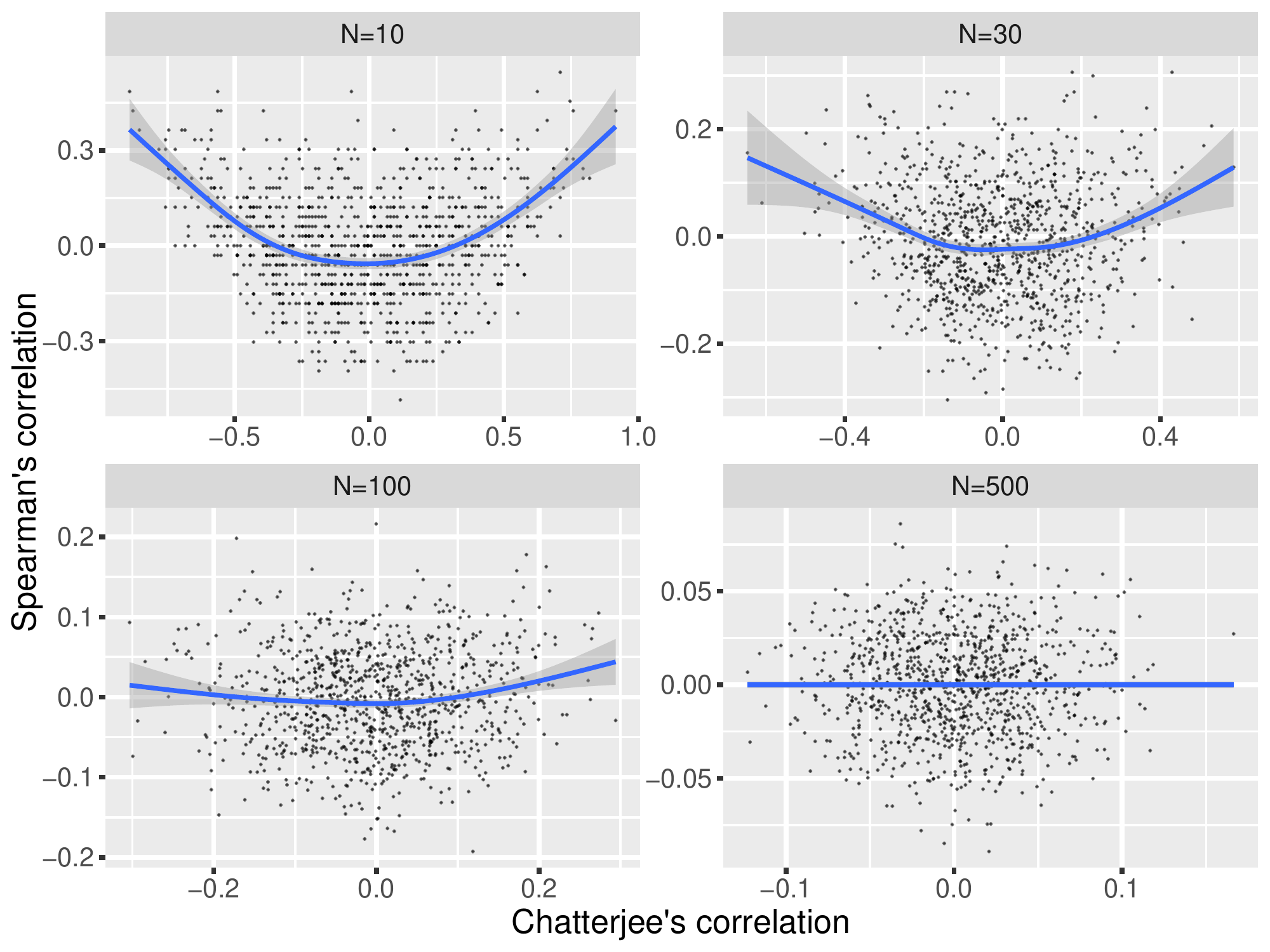}
\end{center}
\caption{Scatterplots of $S_{n}(X, Y)$ and $\xi_{n}(X, Y)$ under $n = 10, 30, 100, 500$.
}
\end{figure}

\newpage
\begin{figure}[!htbp]
\begin{center}
\includegraphics[scale=0.6]{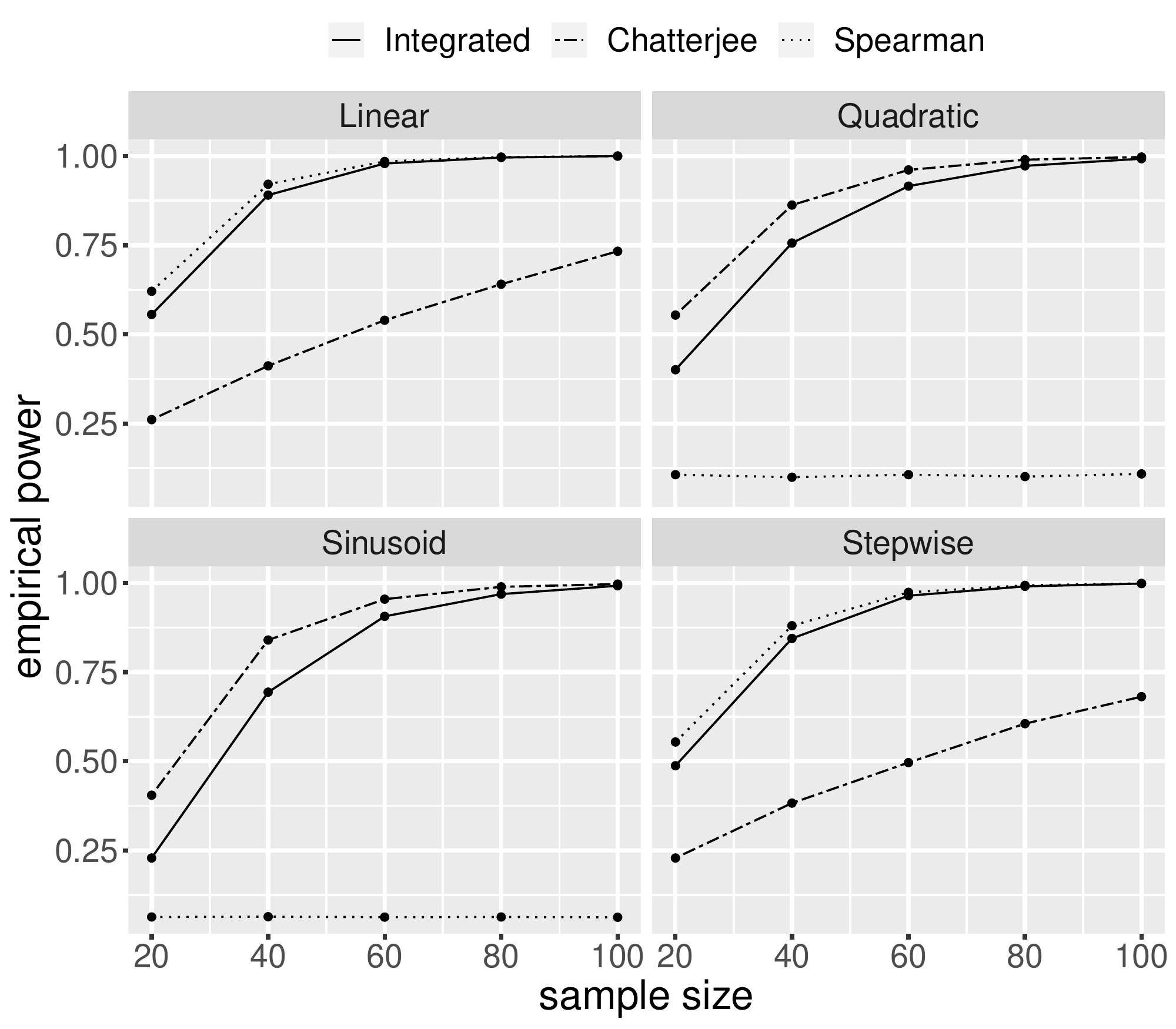}
\end{center}
\caption{Power comparison of Spearman's test, Chatterjee's test and the new test $I_{n}(X, Y)$ under different alternatives and sample sizes.
}
\end{figure}

\newpage
\begin{figure}[!htbp]
\begin{center}
\includegraphics[scale=0.5]{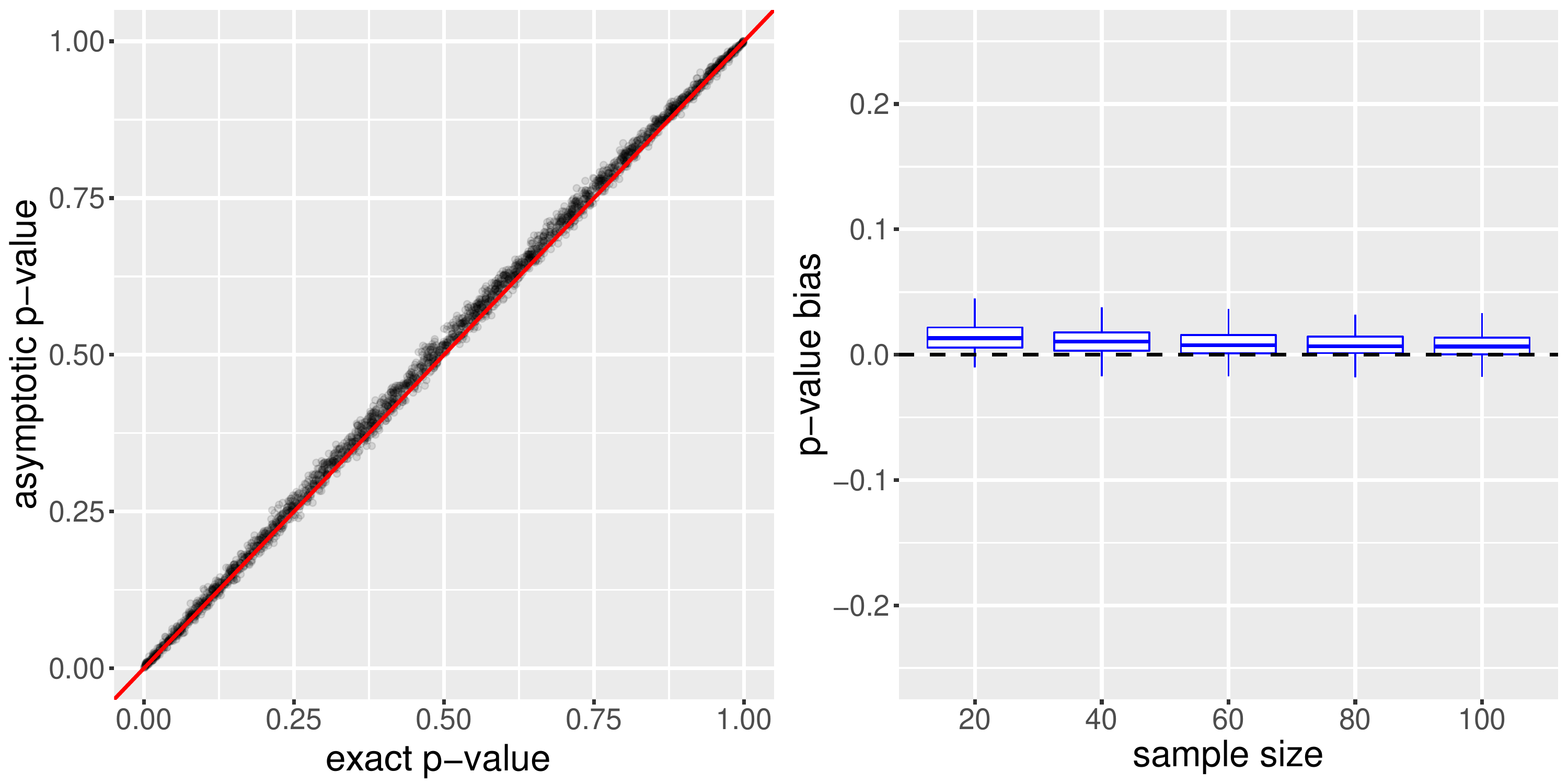}
\end{center}
\caption{Comparison of the asymptotic and exact p-values. Bias is computed as the asymptotic p-value minus the exact p-value.
}
\end{figure}

\newpage
\begin{figure}[!htbp]
\begin{center}
\includegraphics[scale=0.4]{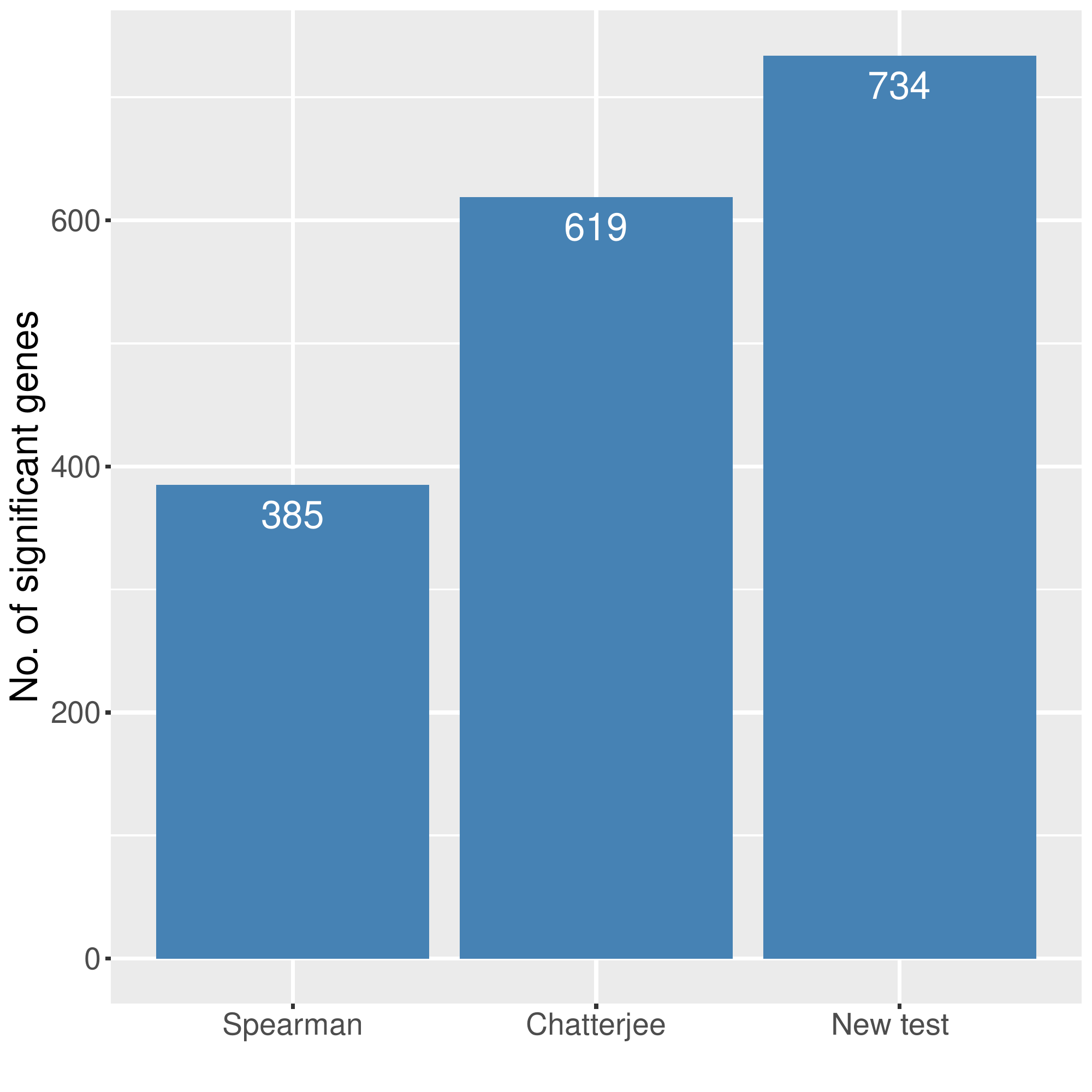}
\end{center}
\caption{Number of significant genes identified by three methods.
}
\end{figure}

\newpage
\begin{figure}[!htbp]
\begin{center}
\includegraphics[scale=0.65]{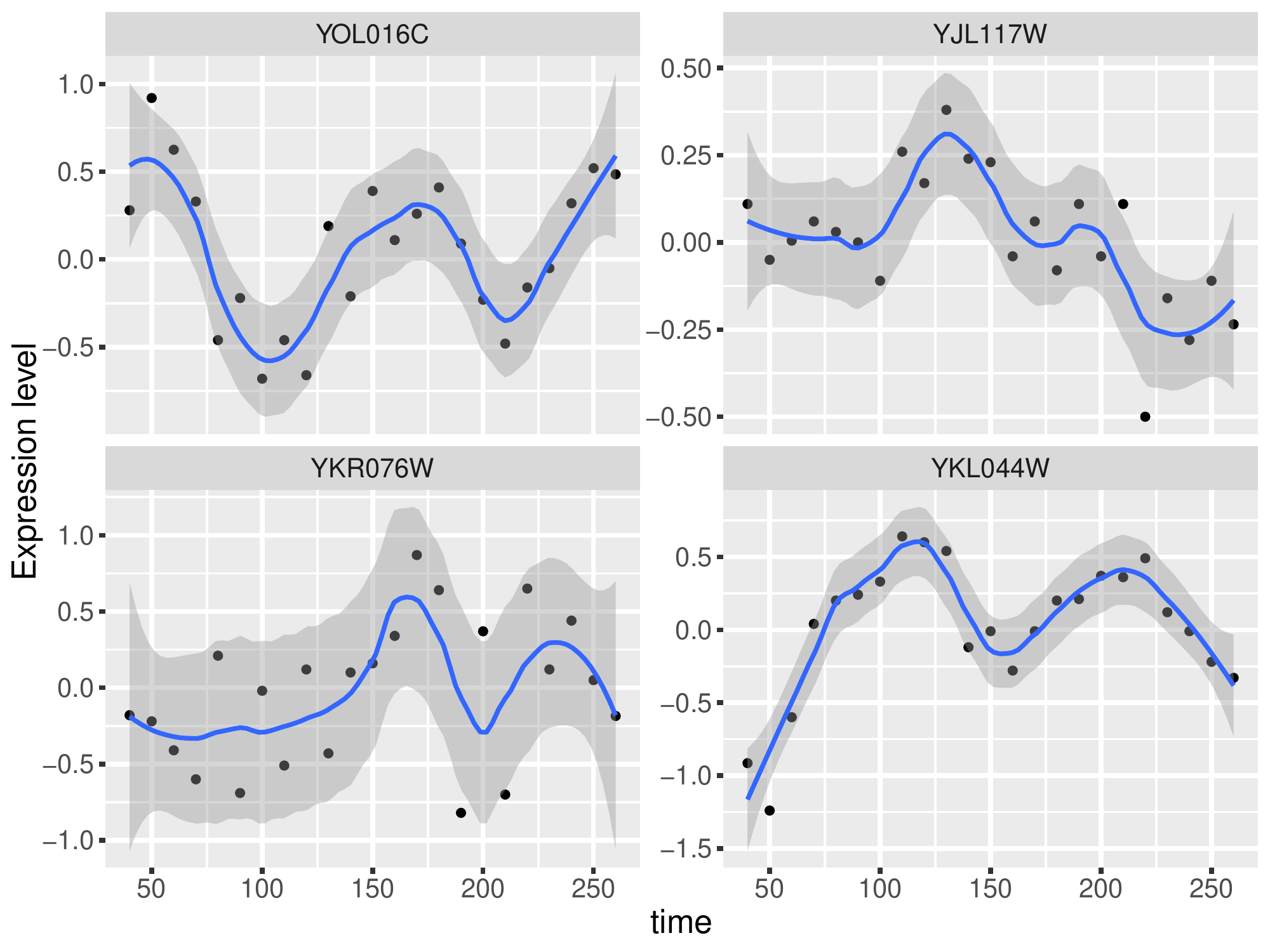}
\end{center}
\caption{A random sample of 4 genes selected by the new test but missed by Spearman's test.
}
\end{figure}

\newpage
\begin{figure}[!htbp]
\begin{center}
\includegraphics[scale=0.65]{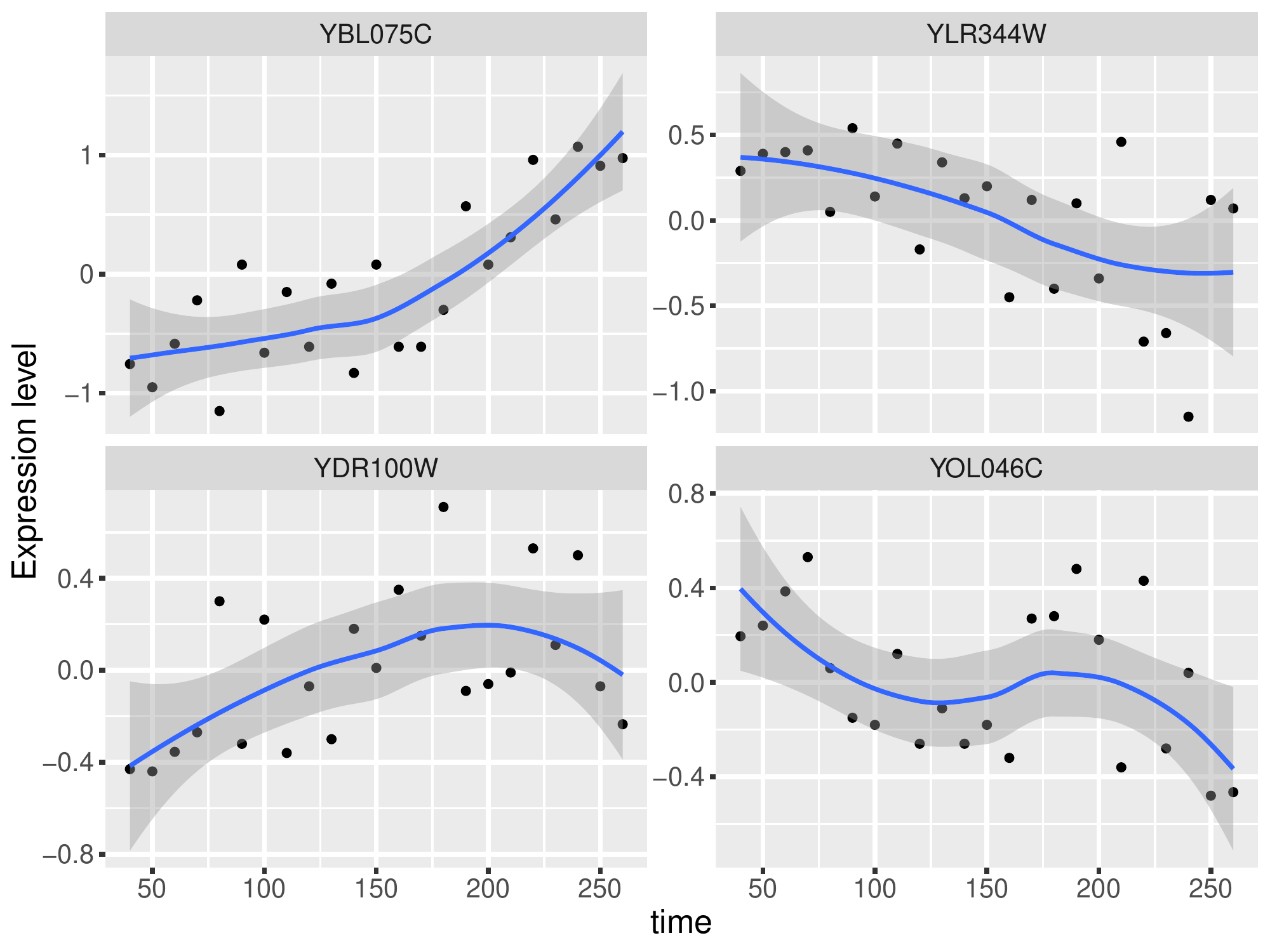}
\end{center}
\caption{A random sample of 4 genes selected by the new test but missed by Chatterjee's test.
}
\end{figure}
\end{document}